%% file: CorenessArxiv_submission.tex
\documentclass[a4paper,12pt]{article}
\pdfoutput=1
\usepackage[totalwidth=500pt,totalheight=680pt]{geometry}
\usepackage{enumerate,amsthm,amsmath,amssymb,mathrsfs}
\usepackage{subfig}
\usepackage{graphicx}
\usepackage{color}
\usepackage[all]{xy}
\usepackage[lined,boxed,ruled]{algorithm2e}

\newtheorem{lemma}{Lemma} 
\newtheorem{theorem}{Theorem} 
\newtheorem{corollary}{Corollary} 
\newtheorem{proposition}{Proposition}

\newtheorem{definition}{Definition} 

\theoremstyle{remark}
\newtheorem{example}{Example} 
 
\newtheorem{remark}{Remark}

\newcommand{\notmodels}{\ \makebox[0.1cm][l]{\ensuremath{\models}}/ \ }

\newcommand{\surhom}{
  \ensuremath{
      \negthinspace 
      \longrightarrow
      \hspace{-5mm} \rightarrow \hspace{1mm}
  }
}

\newenvironment{romannum}[0]
{\begin{enumerate}[(i)]}{\end{enumerate}}

\renewcommand{\phi}{\varphi}

\newcommand{\ignore}[1]{}

\newcommand{\bA}{\mathcal{A}}
\newcommand{\bB}{\mathcal{B}}
\newcommand{\bC}{\mathcal{C}}

\newcommand{\bH}{\mathcal{H}}
\newcommand{\bK}{\mathcal{K}}
\newcommand{\bP}{\mathcal{P}}
\newcommand{\bT}{\mathcal{T}}

\newcommand{\csplogic}{\ensuremath{\{\exists, \wedge \}
    \mbox{-}\mathrm{FO}}}
\newcommand{\qcsplogic}{\ensuremath{\{\exists, \forall, \wedge \}
    \mbox{-}\mathrm{FO}}}
\newcommand{\qcsplogiceq}{\ensuremath{\{\exists, \forall, \wedge,= \} \mbox{-}\mathrm{FO}}}
\newcommand{\mylogic}{\ensuremath{\{\exists, \forall, \wedge,\vee \} \mbox{-}\mathrm{FO}}}
\newcommand{\posFO}{\ensuremath{\{\exists, \forall, \wedge,\vee,= \}
    \mbox{-}\mathrm{FO}}}
\newcommand{\tuple}[1]{\ensuremath{\mathbf{#1}}}
\begin{document}

\title{Containment, Equivalence and Coreness from CSP to QCSP and beyond.}

\author{
Barnaby Martin \\ 
Engineering and Computing Sciences, Durham University, U.K.\thanks{Supported by EPSRC grant EP/G020604/1.} \\
\texttt{barnabymartin@gmail.com}
\and
Florent Madelaine \\ 
Clermont Universit\'{e}, Universit\'{e} d'Auvergne, \\
Clermont-Ferrand, France. \\
\texttt{florent.madelaine@googlemail.com}
}

\maketitle

\begin{abstract}
The constraint satisfaction problem (CSP) and its quantified extensions, 
whether without (QCSP) or with disjunction (QCSP$_\lor$), correspond naturally to the
model checking problem for three increasingly stronger fragments of
positive first-order logic. Their complexity is often studied
when parameterised by a fixed model, the so-called template.
It is a natural question to ask when two templates are equivalent, or
more generally when one ``contain'' another, in the sense that a
satisfied instance of the first will be necessarily satisfied in
the second. One can also ask for a smallest possible equivalent
template: this is known as the core for CSP.
We recall and extend previous results on containment,
equivalence and ``coreness'' for QCSP$_\lor$ before
initiating a preliminary study of cores for QCSP which we characterise for
certain structures and which turns out to be more elusive. 
\end{abstract}

\section{Introduction}

We consider the following increasingly stronger fragments of first-order
logic:
\begin{enumerate}
\item primitive positive first-order ($\csplogic$)
\item positive Horn ($\qcsplogic$)
\item positive equality-free first-order (\mylogic); and,
\item positive first-order logic (\posFO)
\end{enumerate}
The \emph{model checking problem} for a logic $\mathscr{L}$ takes as
input a sentence of $\mathscr{L}$ and a structure $\mathcal{B}$ and
asks whether $\mathcal{B}$ models $\mathcal{L}$. The structure
$\mathcal{B}$ is often assumed to be a fixed parameter and called the
template; and, unless otherwise stated, we will assume
implicitly that we work in this so-called \emph{non-uniform} setting.

For the above first three fragments, the model checking problem is
better known as the \emph{constraint satisfaction problem}
CSP$(\mathcal{B})$, the \emph{quantified constraint satisfaction
  problem} QCSP$(\mathcal{B})$ and its extension with disjunction
which we shall denote by QCSP$_\lor(\mathcal{B})$.
Much of the theoretical research into CSPs is in respect of a large
complexity classification project -- it is conjectured that
CSP$(\mathcal{B})$ is always either in P or NP-complete
\cite{FederVardi}. This \emph{dichotomy} conjecture remains unsettled,
although dichotomy is now known on substantial classes
(e.g. structures of size $\leq 3$ \cite{Schaefer,Bulatov} and smooth
digraphs \cite{HellNesetril,barto:1782}). Various methods,
combinatorial (graph-theoretic), logical and universal-algebraic have
been brought to bear on this classification project, with many
remarkable consequences. A conjectured delineation for the dichotomy
was given in the algebraic language in~\cite{JBK}.

Complexity classifications for QCSPs appear to be harder than for
CSPs. Just as CSP$(\mathcal{B})$ is always in NP, so
QCSP$(\mathcal{B})$ is always in Pspace. No overarching polychotomy
has been conjectured for the complexities of QCSP$(\mathcal{B})$, as
$\mathcal{B}$ ranges over finite structures, but the only known
complexities are P, NP-complete and Pspace-complete (see
\cite{BBCJK,CiE2006} for some trichotomies). It seems plausible that
these complexities are the only ones that can be so obtained. 

Distinct templates may give rise to the same model-checking-problem or preserve acceptance,
\begin{itemize}
\item[] ($\mathscr{L}$-\textbf{equivalence}) for any sentence $\varphi$ of
  $\mathscr{L}$, $\mathcal{A}$ models $\phi$ $\Leftrightarrow$ $\mathcal{B}$
models $\phi$
\item[] ($\mathscr{L}$-\textbf{containment}) for any sentence $\varphi$ of
  $\mathscr{L}$, $\mathcal{A}$ models $\phi$ $\Rightarrow$ $\mathcal{B}$
models $\phi$.
\end{itemize}
We will see that containment and therefore equivalence is decidable,
and often quite effectively so, for the four logics we have introduced. 

For example, when $\mathscr{L}$ is $\csplogic$, any two bipartite
undirected graphs that have at least one edge are equivalent.
Moreover, there is a canonical \emph{minimal} representative for each equivalence
class, the so-called \emph{core}. For example, the core of the class of bipartite
undirected graphs that have at least one edge is the graph $\bK_2$ that
consists of a single edge.
The core enjoys many benign properties
and has greatly facilitated
the classification project for CSPs (which corresponds to the
model-checking for $\csplogic$): it is
unique up to isomorphism 
and sits as an induced substructure in all templates in its
equivalence class. A core may be defined as a structure all of whose
endomorphisms are automorphisms. To review, therefore, it is
well-known that two templates $\mathcal{A}$ and $\mathcal{B}$ are equivalent 
iff there are homomorphisms from $\bA$ to $\bB$ and from $\bB$ to
$\bA$, and in this case there is an (up to isomorphism) unique core
$\bC$ equivalent to both $\mathcal{A}$ and $\mathcal{B}$ such that
$\bC \subseteq \bA$ and $\bC \subseteq \bB$. 

The situation for $\qcsplogic$ and QCSP is somewhat murkier. It is
known that non-trivial $\mathcal{A}$ and $\mathcal{B}$ are equivalent iff there exist integers
$r$ and $r'$ and surjective homomorphisms from $\bA^r$ to $\bB$ and
from $\bB^{r'}$ to $\bA$ (and one may give a bound on these
exponents)~\cite{LICS2008}. However, the status and properties of
``core-ness'' for QCSP were hitherto unstudied. 

We \textbf{might} call a structure $\bB$ a \emph{Q-core} if there is
no equivalent $\bA$ of strictly smaller cardinality. We will discover
that \textbf{this} Q-core is a more cumbersome beast than its cousin
the core; it need not be unique nor sit as an induced substructure of
the templates in its class. However, in many cases we shall see that
its behaviour is reasonable and that -- like the core -- it can be
very useful in delineating complexity classifications. 

The erratic behaviour of Q-cores sits in contrast not just to that of
cores, but also that of the \emph{$U$-$X$-cores} of~\cite{LICS2011},
which are the canonical representatives of the equivalence classes
associated with $\mylogic$,  and were instrumental in deriving a full complexity
classification -- a tetrachotomy -- for QCSP$_\lor$ in \cite{LICS2011}. 
Like cores, they are unique 
and sit as induced substructures in all templates in their
class. Thus, primitive positive logic and positive equality-free logic
behave genially in comparison to their wilder cousin positive Horn. In
fact this manifests on the algebraic side also -- polymorphisms and
surjective hyper-endomorphisms 
are stable under composition, while surjective polymorphisms are not. 

Continuing to add to our logics, in restoring equality, we might
arrive at positive logic. Two finite structures agree on all sentences
of positive logic iff they are isomorphic -- so here every finite
structure satisfies the ideal of ``core''. 
When computing a/the smallest substructure with the same behaviour with
respect to the four decreasingly weaker logics -- positive logic,
positive equality-free, positive Horn, and primitive positive --
we will obtain possibly decreasingly smaller structures. In the case
of positive equality-free and primitive positive logic, as pointed
out, these are unique up to isomorphism; and for the $U$-$X$-core
and the core, these will be induced substructures. \textbf{A} Q-core will
necessarily contain the core and be included in the U-X-core.  
This phenomenon is illustrated on Table~\ref{tab:different-cores} and
will serve as our running example.

\begin{table}[h]
  \centering
  \begin{tabular}[m]{|c|c|c|c|}
    \hline
    \posFO& \mylogic& \qcsplogic& \csplogic\\
    \hline
    $\mathcal{A}_4$&$\mathcal{A}_3$&$\mathcal{A}_2$&$\mathcal{A}_1$\\
    \begin{minipage}[c]{.2\textwidth}
      \centering
      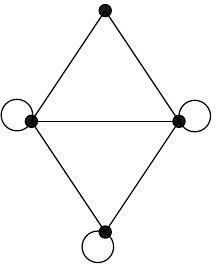
    \end{minipage}
    & 
    \begin{minipage}[c]{.2\textwidth}
      \centering
      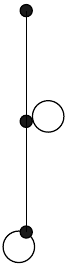
    \end{minipage}
    &
    \begin{minipage}[c]{.2\textwidth}
      \centering
      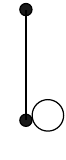
    \end{minipage}
    & 
    \begin{minipage}[c]{.2\textwidth}
            \centering
      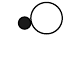
    \end{minipage}
    \\
    &&&\\[-8pt]
    \hline
    isomorphism & $U$-$X$-Core & Q-core & Core\\
    \hline
  \end{tabular}
  \caption{Different notions of "core" (the
circles represent self-loops
).}
  \label{tab:different-cores}
\end{table}

The paper is organised as follows.
In Section~\ref{sec:csp}, we recall folklore results on CSP. In
Section~\ref{sec:mylogic}, we recall results on coreness and spell out containment
for $\mylogic$ that were only implicit in~\cite{LICS2011}. In
Section~\ref{sec:qcsp}, we move on to QCSP and recall results on
the decidability of containment from~\cite{LICS2008} together with new
lower bounds before initiating a study of the notion of core for QCSP.

\section{The case of CSP}
\label{sec:csp}
Unless otherwise stated, we consider structures over a fixed relational
signature $\sigma$. We denote by $A$ the domain of a structure
$\mathcal{A}$ and for every relation symbol $R$ in $\sigma$ of arity
$r$, we write $R^\mathcal{A}$ for the interpretation of $R$ in
$\mathcal{A}$, which is a $r$-ary relation that is
$R^\mathcal{A}\subseteq A^r$. 
We write $|A|$ to denote the cardinality
of the set $A$.
A homomorphism (resp., strong homomorphism) from a structure $\mathcal{A}$ to a structure $\mathcal{B}$ is a function
$h:A\rightarrow B$ such that $(h(a_1),\ldots,h(a_r)) \in R^\mathcal{B}$, if (resp., iff) $(a_1,\ldots,a_r) \in R^\mathcal{A}$.

We will occasionally consider signatures with constant symbols. We
write $c^\mathcal{A}$ for the interpretation of a constant symbol $c$
and homomorphisms are required to preserve constants as well, that is $h(c^\mathcal{A})=c^\mathcal{B}$.

Containment for $\csplogic$ is a special case of conjunctive query
containment from databases~\cite{ChandraMerlin}.
We state and prove these results for pedagogical reasons, before moving to the case of $\mylogic$.
Let us fix some notation first. Given a sentence $\varphi$ in
$\csplogic$, we denote by $\mathcal{D}_\varphi$ its \emph{canonical
  database}, that is the structure with domain the variables of 
$\varphi$ and whose tuples are precisely those that are atoms of
$\varphi$. In the other direction, given a finite structure
$\mathcal{A}$, we write $\phi_{\!\mathcal{A}}$ for the so-called \emph{canonical conjunctive query} of
$\mathcal{A}$, the quantifier-free formula that is the conjunction
of the positive facts of $\mathcal{A}$, where the variables
$v_1,\ldots,v_{|A|}$ correspond to the elements 
$a_1,\ldots,a_{|A|}$ of $\mathcal{A}$.\footnote{Most authors consider the canonical query to be the sentence which is the existential quantification of $\phi_{\!\mathcal{A}}$.} 
It is well known that there is a homomorphism from $\mathcal{D}_\varphi$ to a
structure $\mathcal{A}$ if, and only if,
$\mathcal{A}\models\varphi$. Moreover, a winning strategy for
$\exists$ in the (Hintikka) $(\mathcal{A},\varphi)$-game 
is precisely a homomorphism from $\mathcal{D}_\varphi$ to
$\mathcal{A}$. 
Note also that $\mathcal{A}$ is isomorphic to the canonical 
database of $\exists v_1 \exists v_2 \ldots v_{|A|}
\phi_{\!\mathcal{A}}$.

\begin{theorem}[\textbf{Containment}]
  Let $\mathcal{A}$ and $\mathcal{B}$ be two structures.
  The following are equivalent.
  \begin{romannum}
  \item for every sentence $\varphi$ in $\csplogic$, if $\mathcal{A}\models \varphi$ then
    $\mathcal{B}\models \varphi$.
  \item The exists a homomorphism from $\mathcal{A}$ to $\mathcal{B}$.
  \item $\mathcal{B}\models \exists v_1 \exists v_2 \ldots v_{|A|} \phi_{\!\mathcal{A}}$.
  \end{romannum}
  where $\phi_{\!\mathcal{A}}$ denotes the \emph{canonical conjunctive
    query} of $\mathcal{A}$. 
\end{theorem}
\begin{proof}
  A homomorphism corresponds precisely to a winning strategy in the 
  $(\mathcal{A},\phi)$-game and (ii) and (iii) are
  equivalent. 
  Clearly, (i) implies (iii) since $\mathcal{A}\models  \exists v_1 \exists
  v_2 \ldots v_{|A|} \phi_{\!\mathcal{A}}$.
  
  We now prove that (ii) implies (i). Let $h$ be a homomorphism from
  $\mathcal{A}$ to $\mathcal{B}$. If $\mathcal{A}\models \varphi$,
  then there is a homomorphism $g$ from $\mathcal{D}_\varphi$ to
  $\mathcal{A}$. By composition, $g\circ h$ is a homomorphism from
  $\mathcal{D}_\varphi$ to $\mathcal{B}$. In other words, $g\circ h$
  is a winning strategy witnessing that $\mathcal{B}\models \varphi$.
\end{proof}
It is well known that the core is unique up to isomorphism and that it is
an induced substructure~\cite{HNBook}. It is usually defined via
homomorphic equivalence, but because of the equivalence between (i)
and (ii) in the above theorem, we may define the core as follows.
\begin{definition}
  \textbf{The} \emph{core} $\mathcal{B}$ of a structure $\mathcal{A}$ is a minimal
  substructure of $\mathcal{A}$ such that for every sentence $\varphi$
  in $\csplogic$, $\mathcal{A}\models \varphi$ if and only if
  $\mathcal{B}\models \varphi$. 
\end{definition}
\begin{corollary}[\textbf{equivalence}]
  Let $\mathcal{A}$ and $\mathcal{B}$ be two structures.
  The following are equivalent.
  \begin{romannum}
  \item for every sentence $\varphi$ in $\csplogic$,
    $\mathcal{A}\models \varphi$ if and only if
    $\mathcal{B}\models \varphi$.
  \item There are homomorphisms from $\mathcal{A}$ to $\mathcal{B}$
    and from $\mathcal{B}$ to $\mathcal{A}$.
  \item The core of $\mathcal{A}$ and the core of $\mathcal{B}$ are isomorphic.
  \end{romannum}
\end{corollary}
As a preprocessing step, one could
replace the template $\mathcal{A}$ of a CSP by its core
$\mathcal{B}$ (see Algorithm 6.1 in~\cite{CohenJeavonsSurvey}). 
However,
the complexity of this preprocessing step would be of the same order of
magnitude as solving a constraint satisfaction problem.\footnotemark{}
This drawback, together with the uniform nature of the instance in
constraints solvers, means that this preprocessing is not exploited
in practice to the best of our knowledge.
\footnotetext{Checking that a graph is a core is
  coNP-complete~\cite{cores}. Checking that a graph is the core of
  another given graph is
  DP-complete~\cite{DBLP:journals/tods/FaginKP05}.}

The notion of a core can be extended and adapted suitably to solve
important questions related to data exchange and query rewriting in
databases~\cite{DBLP:journals/tods/FaginKP05}. 
It is also very useful as a simplifying assumption when classifying
the complexity: with the algebraic approach, it allows to study only
idempotent algebras~\cite{JBK}. 

\section{The case of QCSP with disjunction}
\label{sec:mylogic}

For $\mylogic$, it is no longer the homomorphism that is the correct concept to transfer winning
strategies. 
\begin{definition}
  A \emph{surjective hypermorphism} $f$ from a structure
  $\mathcal{A}$ to a structure $\mathcal{B}$ is a
  function from the domain $A$ of $\mathcal{A}$ to the power set
  of the domain $B$ of $\mathcal{B}$ that satisfies the following properties.
  \begin{itemize}
  \item (\textbf{total}) for any $a$ in $A$, $f(a)\neq \emptyset$.
  \item (\textbf{surjective}) for any $b$ in $B$, there exists $a$ in
    $A$ such that $f(a)\ni b$.
  \item (\textbf{preserving}) if $R(a_1,\ldots,a_i)$ holds in
    $\mathcal{A}$ then $R(b_1,\ldots,b_i)$ holds in $\mathcal{B}$ , for all $b_1 \in f(a_1),\ldots,$ $b_i \in f(a_i)$. 
  \end{itemize}
\end{definition}

A \emph{strategy} for $\exists$ in the (Hintikka) $(\mathcal{A},\phi)$-game, where $\varphi \in \mylogic$, is 
a set of mappings $\{\sigma_x : \mbox{`$\exists x$'} \in \phi \}$ with one mapping $\sigma_x$ for each existentially
quantified variable $x$ of $\varphi$.
The mapping $\sigma_x$ ranges over the domain $A$ of $\mathcal{A}$; and, its domain is the
set of functions from $Y_x$ to $A$, where $Y_x$ denotes the
universally quantified variables of $\phi$ preceding $x$.

We say that $\{\sigma_x : \mbox{`$\exists x$'} \in \phi \}$ is \emph{winning} if for any assignment
$\pi$ of the universally quantified variables of $\varphi$ to $A$, when each 
existentially quantified variable $x$ is set according to $\sigma_x$
applied to $\left.\pi\right|_{Y_x}$, then the quantifier-free part $\psi$ of
$\varphi$ is satisfied under this overall assignment $h$.
When $\psi$ is disjunction-free, this amounts to $h$ being a
homomorphism from $\mathcal{D}_\psi$ to $\mathcal{A}$.

\begin{lemma}[\textbf{strategy transfer}]\label{lemma:strategy:transfert:mylogic}
  Let $\mathcal{A}$ and $\mathcal{B}$ be two structures such that
  there is a surjective hypermorphism from $\mathcal{A}$ to
  $\mathcal{B}$. Then, for every sentence $\varphi$ in $\mylogic$, if
  $\mathcal{A}\models \varphi$ then $\mathcal{B}\models \varphi$.
\end{lemma}
\begin{proof}
  Let $f$ be a surjective hypermorphism from $\mathcal{A}$ to
  $\mathcal{B}$ and $\varphi$ be a sentence of $\mylogic$ such that
  $\mathcal{A}\models \varphi$.   For any element $b$ of
  $\mathcal{B}$, let $f^{-1}(b):=\{a \in A \mbox{ s. t. } b\in f(a)\}$.
  We fix an arbitrary linear order over $A$ and write $\min f^{-1}(b)$
  to denote the smallest antecedent of $b$ in $A$ under $f$. 
  
  Let $\{\sigma_x : \mbox{`$\exists x$'} \in \phi \}$ be a winning strategy in the 
  $(\mathcal{A},\phi)$-game. We construct a strategy $\{\sigma'_x : \mbox{`$\exists x$'} \in \phi \}$
  in the 
  $(\mathcal{B},\phi)$-game as follows.
  Let $\pi_B:Y_x\to B$ be an assignment to the universal variables $Y_x$
  preceding an existential variable $x$ in $\varphi$, we select for 
  $\sigma'_x(\pi)$ an arbitrary element of $f(\sigma(\pi_A))$ where
  $\pi_A:Y_x\to A$ is an assignment such that for any universal variable
  $y$ preceding $x$, we have $\pi_A(y):=\min f^{-1}(\pi_B(y))$. This strategy is well
  defined since $f$ is surjective (which means that $\pi_A$ is well
  defined) and total (which means that $f(\sigma(\pi_A))\neq
  \emptyset$). Note moreover that using $\min$ in the definition of
  $\pi_A$ means that a branch in the tree of the game on $\mathcal{B}$ will
  correspond to a branch in  the tree of the game on $\mathcal{A}$.
  It remains to prove that $\{\sigma'_x : \mbox{`$\exists x$'} \in \phi \}$ is winning. We will see
  that it follows from the fact that $f$ is preserving.

  Assume first that $\varphi$ is a sentence of $\qcsplogic$.
  Let $\mathcal{D}_\psi$ be the canonical
  database of the quantifier-free part $\psi$ of $\varphi$. The winning
  condition of the 
  $(\mathcal{B},\phi)$-game can be recast as a homomorphism from
  $\mathcal{D}_\psi$. Composing with $f$ the homomorphism from $\mathcal{D}_\psi$ to
  $\mathcal{A}$ (induced by the sequence of compatible assignments
  $\pi_A$ to the universal variables and the strategy
  $\{\sigma_x : \mbox{`$\exists x$'} \in \phi \}$), we get a surjective hypermorphism from $\mathcal{D}_\psi$ to
  $\mathcal{B}$. The map from the domain of $\mathcal{D}_\psi$ to $\mathcal{B}$
  induced by the sequence of assignments $\pi_B$ and the strategy
  $\{\sigma'_x : \mbox{`$\exists x$'} \in \phi \}$ is a range restriction of this surjective
  hypermorphism and is therefore a homomorphism (we identify
  surjective hypermorphism to singletons with homomorphisms).

  When $\varphi$ is not a sentence of $\qcsplogic$, we write its
  quantifier-free part in  disjunctive normal
  form as a disjunction of conjunctions-of-atoms
  $\psi_i$. The winning condition can now be recast as a homomorphism
  from some $\mathcal{D}_{\psi_i}$. The above argument applies and the result follows.
\end{proof}
\begin{example}
  Consider the structures $\mathcal{A}_4$ and  $\mathcal{A}_3$ from
  Table~\ref{tab:different-cores}.
  The map  $f$ given by $f(1):=\{1\}, f(2):=\{2\}, f(3):=\{3\},
  f(4):=\{1\}$ is a surjective hypermorphism from  $\mathcal{A}_4$ to
  $\mathcal{A}_3$. The map $g$ given by $g(1):=\{1,4\}, g(2):=\{2\},
  g(3):=\{3\}$ is a surjective hypermorphism from  $\mathcal{A}_3$ to 
  $\mathcal{A}_4$. The two templates are equivalent w.r.t. $\mylogic$.
\end{example}
We extend the notion of canonical conjunctive query of a structure $\mathcal{A}$. Given a tuple
of (not necessarily distinct) elements $\tuple{r}:=(r_1,\ldots,r_l)
\in A^l$, define the quantifier-free formula
$\phi_{\!\mathcal{A}(\tuple{r})}(v_1,\ldots,v_l)$ to be the conjunction
of the positive facts of $\tuple{r}$, where the variables
$v_1,\ldots,v_l$ correspond to the elements $r_1,\ldots,r_l$. That is,
$R(v_{\lambda_1},\ldots,v_{\lambda_i})$ appears as an atom in
$\phi_{\!\mathcal{A}(\tuple{r})}$ iff
$R(r_{\lambda_1},\ldots,r_{\lambda_i})$ holds in $\mathcal{A}$. 
When $\tuple{r}$ enumerates the elements of the structure
$\mathcal{A}$, this definition coincides with the usual definition of canonical
conjunctive query.
Note also that there is a strong homomorphism from the canonical database
$\mathcal{D}_{\phi_\mathcal{A}(\tuple{r})}$ to $\mathcal{A}$ given by
the map $r_i \mapsto v_i$.

\begin{definition}[\textbf{Canonical $\mylogic$ sentence}]
  Let $\mathcal{A}$ be a structure and $m>0$.
  Let $\tuple{r}$ be an enumeration of the elements of $\mathcal{A}$.
  $$\theta_{\!\mathcal{A},m}:=
  \exists v_1, \ldots, v_{|A|}
  \phi_{\!\mathcal{A}(\tuple{r})}(v_1,\ldots,v_{|A|})
  \land
  \forall w_1,\ldots,w_m 
  \bigvee_{\tuple{t} \in A^{m}}
  \phi_{\!\mathcal{A}(\tuple{r},\tuple{t})}(\tuple{v},\tuple{w}).
  $$
\end{definition}
Observe that $\mathcal{A}\models \theta_{\!\mathcal{A},m}$. Indeed, we
may take as  witness for the variables $\tuple{v}$ the corresponding enumeration $\tuple{a}$
of the elements of
$\mathcal{A}$; and, for any assignment $\tuple{t} \in A^{m}$ to the
universal variables $\tuple{w}$, it is clear that
$\mathcal{A}\models\phi_{\!\mathcal{A}(\tuple{r},\tuple{t})}(\tuple{a},\tuple{t})$ holds.
\begin{lemma}[strategy transfer]\label{lemma:canonical:sentence:mylogic}
  Let $\mathcal{A}$ and $\mathcal{B}$ be two structures.
  If $\mathcal{B}\models \theta_{\!\mathcal{A},|\mathcal{B}|}$ then
  there is a surjective hypermorphism from $\mathcal{A}$ to $\mathcal{B}$. 
\end{lemma}
\begin{proof}
  Let $\tuple{b'}:=b'_1,\ldots,b'_{|A|}$ be witnesses for $v_1,\ldots,v_{|A|}$.
  Assume that an enumeration $\tuple{b}:=b_1,b_2,\ldots,b_{|B|}$ of the elements of
  $\mathcal{B}$  is chosen for the
  universal variables $w_1,\ldots w_{|\mathcal{B}|}$.
  Let $\tuple{t} \in A^{m}$ be the witness s.t. 
  $\mathcal{B}\models 
  \phi_{\!\mathcal{A}(\tuple{r})}(\tuple{b'})
  \land
  \phi_{\!\mathcal{A}(\tuple{r},\tuple{t})}(\tuple{b'},\tuple{b})
  $.
  
  Let $f$ be the map from the domain of $\mathcal{A}$ to the power set
  of that of $\mathcal{B}$ which is the union of the following two
  partial hyperoperations $h$ and $g$ (\textsl{i.e.} $f(a_i):=h(a_i)\cup g(a_i)$
  for any element $a_i$ of $\mathcal{A}$), which guarantee totality
  and surjectivity, respectively.
  \begin{itemize}
  \item (\textbf{totality}) $h(a_i):=b'_i$
  \item (\textbf{surjectivity}) $g(t_i)\ni b_i$.
  \end{itemize}
  It remains to show that $f$ is preserving. This follows from
  $\mathcal{B}\models
  \phi_{\!\mathcal{A}(\tuple{r},\tuple{t})}(\tuple{b'},\tuple{b})$. 

  Let $R$ be a $r$-ary relational symbol such that
  $R(a_{i_1},\ldots,a_{i_r})$ holds in $\mathcal{A}$. Let $b''_{i_1}\in f(a_{i_1}), \ldots, b''_{i_r}\in
  f(a_r)$. We will show that $R(b''_{i_1},\ldots,b''_{i_r})$ holds in $\mathcal{B}$.
  Assume for clarity of the exposition and w.l.o.g. that from $i_1$ to $i_k$ the image is set according to
  $h$ and from $i_{k+1}$ to $i_r$ according to $g$:
  i.e. for $1\leq j \leq k$, $h(a_{i_j})=b'_{i_j}=b''_{i_j}$  and 
  for $k+1\leq j \leq r$, there is some $l_j$ such that
  $t_{l_j}=a_{i_j}$ and $g(t_{l_j})\ni b''_{i_j}=b_{l_j}$. 
  By definition of  $\mathcal{A}(\tuple{r},\tuple{t})$ the atom
  $R(v_{i_1},\ldots,v_{i_k},w_{l_{k+1}},\ldots,w_{r})$ appears  in  $\phi_{\!\mathcal{A}(\tuple{r},\tuple{t})}(\tuple{v},\tuple{w})$. 
  It follows from   $\mathcal{B}\models
  \phi_{\!\mathcal{A}(\tuple{r},\tuple{t})}(\tuple{b'},\tuple{b})$
  that $R(b''_{i_1},\ldots,b''_{i_r})$ holds in $\mathcal{B}$.
\end{proof}
\begin{theorem}[\textbf{Containment for \mylogic}]
  \label{theo:containment:mylogic}
  Let $\mathcal{A}$ and $\mathcal{B}$ be two structures.
  The following are equivalent.
  \begin{romannum}
  \item for every sentence $\varphi$ in $\mylogic$, if $\mathcal{A}\models \varphi$ then
    $\mathcal{B}\models \varphi$.
  \item The exists a surjective hypermorphism from $\mathcal{A}$ to $\mathcal{B}$.
  \item $\mathcal{B}\models \theta_{\mathcal{A},|B|}$
  \end{romannum}
  where $\Theta_{\mathcal{A,B}}$ is a canonical sentence of $\mylogic$ that is
  defined in terms of $\mathcal{A}$ and $|\mathcal{B}|$ and that
  is modelled by $\mathcal{A}$ by construction.
\end{theorem}
\begin{proof}
  By construction $\mathcal{A}\models \theta_{\mathcal{A},|B|}$, so
  (i) implies (iii). By Lemma~\ref{lemma:strategy:transfert:mylogic}, (ii) implies (i).
  By Lemma~\ref{lemma:canonical:sentence:mylogic}, (iii) implies (i).
\end{proof}
Let $U$ and $X$ be two subsets of $A$ and a surjective hypermorphism $h$ from
$\mathcal{A}$ to $\mathcal{A}$ 
that satisfies $h(U)=A$ and
$h^{-1}(X)=A$. Let $\mathcal{B}$ be the substructure of $\mathcal{A}$
induced by $B:=U\cup X$. Then $f$ and $g$, the range and domain restriction of $h$ to
$B$, respectively, are surjective hypermorphisms between $\mathcal{A}$
and $\mathcal{B}$ witnessing that $\mathcal{A}$ and $\mathcal{B}$
satisfy the same sentence of $\mylogic$. 
Note that in particular $h$ induces a retraction of $\mathcal{A}$ to a
subset of $X$; and, dually a retraction of the complement structure~\footnotemark{} of
$\mathcal{B}$ to a subset of $U$.
\footnotetext{It has the same domain as $\mathcal{A}$ and a tuple belongs to a relation
$R$ iff it did not in $\mathcal{A}$.}
Additional minimality conditions on $U$, $X$ and $U\cup
X$ ensure that $\mathcal{B}$ is minimal.\footnotemark{}
\footnotetext{This is possible since given $h_1$ s.t. $h_1(U)=A$ and
  $h_2$ such that $h_2^{-1}(X)=A$, their composition $h=h_2\circ h_1$
  satisfies both $h(U)=A$ and
  $h^{-1}(X)=A$.}
It is also \emph{unique} up to isomorphism and within $\mathcal{B}$ the set
$U$ and $X$ are \emph{uniquely determined}. Consequently, $\mathcal{B}$ is
called \textbf{the} $U$-$X$-core of $\mathcal{A}$ (for further details
see~\cite{LICS2011}) and may be defined as follows.

\begin{definition}
  \textbf{The} $U$-$X$\emph{core} $\mathcal{B}$ of a structure $\mathcal{A}$ is a minimal
  substructure of $\mathcal{A}$ such that for every sentence $\varphi$
  in $\mylogic$, $\mathcal{A}\models \varphi$ if and only if
  $\mathcal{B}\models \varphi$. 
\end{definition}

\begin{example}
  The map $h(1):=\{1,4\}, h(2):=\{2\}, h(3):=\{1,3,4\}, h(4):=\{1,4\}$
  is a surjective hypermorphism from $\mathcal{A}_4$ to
  $\mathcal{A}_4$ with $U=\{2,3\}$ and $X:=\{1,2\}$.
  The substructure induced by $U\cup X$ is $\mathcal{A}_3$.
  It can be checked that it is minimal.
\end{example}

The $U$-$X$-core is just like the core an induced substructure. There
is one important difference in that $U$-$X$-cores should be
genuinely viewed as a minimal equivalent substructure induced by \emph{two}
sets. Indeed, when evaluating a sentence of $\mylogic$, we may
assume w.l.o.g. that all $\forall$ variables range over $U$ and all $\exists$
variables range over $X$.
This is because for any play of $\forall$, we may extract a winning
strategy for $\exists$ that can even restrict herself to play only on $X$~\cite[Lemma 5]{LICS2011}.
Hence, as a \emph{preprocessing step}, one could compute $U$ and $X$ and
restrict the domain of each universal variable to $U$ and the domain
of each universal variable to $X$. 
\emph{The complexity of this processing step is no longer of the same
magnitude} and is in general much lower than solving a QCSP$_\lor$.\footnotemark{}
\footnotetext{The question of $U$-$X$-core identification is in DP (and should be complete), whereas QCSP$_\lor$ is Pspace-complete in general
}
Thus, even when taking into account the uniform nature of the instance in
a quantified constraints solver, this preprocessing step might be
exploited in practice. This could turn out to be ineffective when there are few
quantifier alternation (as in bilevel programming), but should be of
particular interest when the quantifier alternation increases. 
Another
interesting feature is that storing a winning strategy over $U$ and
$X$ together with the surjective hypermorphism $h$ from
$\mathcal{A}$ to $\mathcal{A}$, allows to recover a winning strategy
even when $\forall$ plays in an unrestricted manner. This provides a
\emph{compression mechanism} to store certificates.

\section{The case of QCSP}
\label{sec:qcsp}

In primitive positive and positive Horn logic, one normally considers equalities to be permitted. From the perspective of computational complexity of CSP and QCSP, this distinction is unimportant as equalities may be propagated out by substitution. In the case of positive Horn and QCSP, though, equality does allow the distinction of a trivial case that can not be recognised without it. The sentence $\forall x \ x=x$ is true exactly on structures of size one. The structures $\mathcal{K}_1$ and $2 \mathcal{K}_1$, containing empty relations over one element and two elements, respectively, are therefore distinguishable in $\qcsplogiceq$, but not in $\qcsplogic$. Since we disallow equalities, many results from this section apply only to \emph{non-trivial} structures of size $\geq 2$. Note that equalities can not be substituted out from $\posFO$, thus it is substantially stronger than $\mylogic$.

For $\qcsplogic$, the correct concept to transfer winning strategies
is that of \emph{surjective homomorphism from a power}.
Recall first that the \emph{product} $\mathcal{A} \times \mathcal{B}$ of two
structures $\mathcal{A}$ and $\mathcal{B}$ has domain $\{(x,y):x \in
A, y \in B\}$ and for a relation symbol $R$,
$R^{\mathcal{A}\times\mathcal{B}}:=\{\bigl((a_1,b_1),\ldots,(a_r,b_r)\bigr): (a_1,\ldots,a_r) \in R^\mathcal{A}, (b_1,\ldots,b_r) \in
R^\mathcal{B}\}$; and, similarly for a constant symbol $c$, $c^{\mathcal{A}\times\mathcal{B}}:=(c^{\mathcal{A}},c^{\mathcal{B}})$.
The \emph{$m$th power} $\mathcal{A}^m$ of $\mathcal{A}$ is
$\mathcal{A} \times \ldots \times \mathcal{A}$ ($m$ times).

\begin{lemma}[\textbf{strategy transfer}]\label{lemma:strategy:transfert:qcsplogic}
  Let $\mathcal{A}$ and $\mathcal{B}$ be two structures and $m\geq 1$ such that
  there is a surjective homomorphism from $\mathcal{A}^m$ to
  $\mathcal{B}$. Then, for every sentence $\varphi$ in $\qcsplogic$, if
  $\mathcal{A}\models \varphi$ then $\mathcal{B}\models \varphi$.
\end{lemma}
\begin{proof}
  For $m=1$, the proof is similar to
  Lemma~\ref{lemma:strategy:transfert:mylogic}.
  A projection from $\mathcal{A}^m$ to $\mathcal{A}$ is a surjective
  homomorphism. This means that for every sentence $\varphi$ in $\qcsplogic$, if
  $\mathcal{A}^r\models \varphi$ then $\mathcal{A}\models \varphi$.
  For the converse, one can consider the ``product strategy'' which consists
  in projecting over each coordinate of $\mathcal{A}^m$ and applying the
  strategy for $\mathcal{A}$. For further details see~\cite[Lemma 1\&2]{LICS2008}.
\end{proof}
\begin{example}\label{ex:bip}
  Consider an undirected bipartite graphs with at least one edge
  $\mathcal{G}$ and $\bK_2$ the graph that consists of a single edge.
  There is a surjective homomorphism from $\mathcal{G}$ to $\bK_2$.
  Note also that $\bK_2\times \bK_2 = \bK_2+\bK_2$ (where $+$ stands for disjoint union)
  which we write as $2\bK_2$. Thus, ${\bK_2}^j=2^{j-1}\bK_2$ (as $\times$
  distributes over $+$). Hence, if $\mathcal{G}$ has no isolated element and
  $m$ edges there is a surjective homomorphism from $\bK_2^{1+\log_2
    m}$ to $\mathcal{G}$.
\end{example}
This examples provides a lower bound for $m$ which we can improve.
\begin{proposition}[lower bound]
  \label{prop:lowerbound:powerepimorphism}
  For any $m\geq 2$,
  there are structures $\mathcal{A}$ and $\mathcal{B}$ with 
  $|A|=m$ and $|B|=m+1$
  such that there
  is only a surjective homomorphism from $\mathcal{A}^j$ to $\mathcal{B}$
  provided that $j \geq |A|$.
\end{proposition}
\begin{proof}[sketch]
  \begin{figure}[h]
    \centering
    \includegraphics{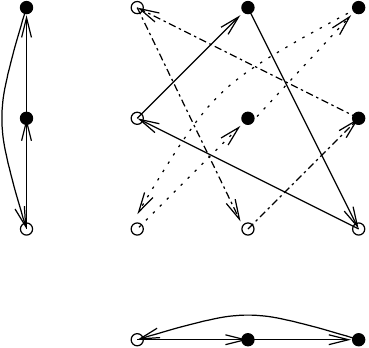}
    \caption{The power of oriented cycles is a sum of oriented cycles.}
    \label{fig:power:cycles}
  \end{figure}
  We consider a signature that consists of a binary symbol $E$
  together with a monadic predicate $R$.
  Consider for $\mathcal{A}$ an oriented cycle with $m$ vertices, for
  which $R$ holds for all but one. Consider for $\mathcal{B}$ an oriented
  cycle with $m$ vertices, for which $R$ does not hold, together with
  a self-loop on which $R$ holds.
  The square of $\mathcal{A}$ will consists of $|A|=m$ oriented cycles
  with $m$ vertices: one cycle will be a copy of $\mathcal{A}$, all
  the other will be similar but with two vertices on which $R$ does
  not hold (this is depicted on Figure~\ref{fig:power:cycles} in the
  case $m=3$: white
  vertices do not satisfy $R$ while black ones do).
  It is only for $j=m$ that we will get as an
  induced substructure of $\mathcal{A}^j$ one copy of an oriented cycle on which $R$ does
  not hold as in $\mathcal{B}$.
\end{proof}
There is also a \emph{canonical $\qcsplogic$-sentence} which turns out
to be in $\Pi_2$-form, that is with a quantifier prefix of the form
$\forall^\star\exists^\star$. We consider temporarily structures with
$m$ constants $c_1,c_2,\ldots,c_m$; let $\tuple{t}$ in $A^{m}$
describe the position of these constants in a structure $\mathcal{A}$;
and, write $\mathcal{A}_\tuple{t}$ for the corresponding structure with
constants. 
We consider the canonical
conjunctive query of the structure with constants
$\bigotimes_{\tuple{t}\in A^{m}}\mathcal{A}_\tuple{t}$, 
(where $\bigotimes$ denote the product), identifying the constants
with some variables $\tuple{w}=w_1,\ldots,w_m$ and using variables
$\tuple{v}$ for the other elements. We turn this quantifier-free
formula into a sentence by adding the prefix
$\forall\tuple{w}\exists\tuple{v}$.
Keeping this in mind, we can also give the following direct
definition, but it dilutes the intuition somewhat.
\begin{definition}[\textbf{Canonical $\qcsplogic$ sentence}]
  \label{def:canonical:pi2:sentence}
  Let $\mathcal{A}$ be a structure and $m>0$.
  Let $\tuple{r}$ be an enumeration of the elements of $\tilde{\mathcal{A}}:=\mathcal{A}^{|A|^m}$.
  $$\psi_{\!\mathcal{A},m}:=
  \forall \tuple{w}
  \exists \tuple{v}
  \phi_{\!\tilde{\mathcal{A}}(\tuple{r})}(\tuple{v})
  \land
  \bigwedge_{\tuple{t} \in A^{m}}
  w_1=v_{\tuple{t},\tuple{t}[1]} \ldots \land  w_m=v_{\tuple{t},\tuple{t}[m]}
  .
  $$
\end{definition}
Observe that we may propagate the equalities out of
$\psi_{\!\mathcal{A},m}$ to obtain an equivalent sentence:
e.g. we remove $w_1=v_{\tuple{t},\tuple{t}[1]}$ and replace every
occurrence of $v_{\tuple{t},\tuple{t}[1]}$ by $w_1$. 
Observe also that $\mathcal{A}\models \psi_{\!\mathcal{A},m}$.
Indeed, assume that $\tuple{t}\in A^{m}$ is the assignment chosen for the
universal variables $\tuple{w}$. There is a natural projection from
$\bigotimes_{\tuple{t}\in A^{m}}\mathcal{A}_\tuple{t}$ to
$\mathcal{A}_\tuple{t}$ which is a homomorphism. This homomorphism
corresponds precisely to a winning strategy for the existential variables $\tuple{v}$. 
\begin{theorem}[\textbf{Containment for \qcsplogic}]
  \label{theo:containment:qcsplogic}
  Let $\mathcal{A}$ and $\mathcal{B}$ be two non-trivial structures.
  The following are equivalent.
  \begin{romannum}
  \item for every sentence $\varphi$ in $\qcsplogic$, if $\mathcal{A}\models \varphi$ then
    $\mathcal{B}\models \varphi$.
  \item There exists a surjective homomorphism from $\mathcal{A}^r$ to
    $\mathcal{B}$, with $r\leq |A|^{|B|}$.
  \item $\mathcal{B}\models \psi_{\mathcal{A},|B|}$
  \end{romannum}
  where $\psi_{\mathcal{A},|\mathcal{B}|}$ is a canonical sentence of
  $\qcsplogic$ with quantifier prefix $\forall^{|\mathcal{B}|}\exists^\star$ that is
  defined in terms of $\mathcal{A}$ and modelled by $\mathcal{A}$ by construction.
\end{theorem}
\begin{proof}[sketch]
  (ii) implies (i) by Lemma~\ref{lemma:strategy:transfert:qcsplogic}.
  (i) implies (iii) since $\mathcal{A}$ models
  $\psi_{\mathcal{A},|\mathcal{B}|}$.
  (iii) implies (ii) by construction of
  $\psi_{\mathcal{A},|\mathcal{B}|}$. We may chose for the universal
  variables $\tuple{w}$ an enumeration of $\mathcal{B}$. The winning
  strategy on $\mathcal{B}$ induces a surjective homomorphism from
  $\mathcal{A}^r$ (for further details see~\cite[Theorem 3]{LICS2008}
  and comments on the following page). 
\end{proof}
Following our approach for the other logics, we now define a minimal
representative as follows.
\begin{definition}
  \textbf{A} \emph{Q-core} $\mathcal{B}$ of a structure $\mathcal{A}$ is a minimal
  substructure of $\mathcal{A}$ such that for every sentence $\varphi$
  in $\qcsplogic$, $\mathcal{A}\models \varphi$ if and only if
  $\mathcal{B}\models \varphi$. 
\end{definition}
\begin{figure}[h]
  \centering
  \subfloat[$\mathcal{A}_2\times \mathcal{A}_2$.]{\label{fig:squareA2}
    \begin{minipage}[c]{.3\textwidth}
      \centering
      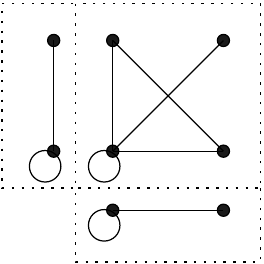
    \end{minipage}
}
  \qquad
  \subfloat[Homomorphism to $\mathcal{A}_3$.]{\label{fig:surjhom:squareA2:A3}
    \begin{minipage}[c]{.3\textwidth}
      \centering
      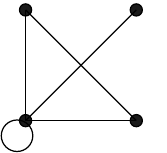
    \end{minipage}
}
  \qquad  
  \subfloat[$\mathcal{A}_3$]{\label{fig:A3}
    \begin{minipage}[c]{.2\textwidth}
      \centering
      \input{DifCores3.pdf_t}
    \end{minipage}
}
  \qquad
  \caption{Surjective homomorphism from a power.}
  \label{fig:squareA3}
\end{figure}

\begin{example}
  Consider $\mathcal{A}_3$ and $\mathcal{A}_2$ from Table~\ref{tab:different-cores}.
  The map $f(1):=1,f(2):=2, f(3):=2$ is a surjective homomorphism from
  $\mathcal{A}_3$ to $\mathcal{A}_2$. 
  The square of $\mathcal{A}_2$ is depicted on Figure~\ref{fig:squareA2};
  and, a surjective homomorphism from it to $\mathcal{A}_3$ is
  depicted on Figure~\ref{fig:surjhom:squareA2:A3}.
  Thus $\mathcal{A}_3$ and $\mathcal{A}_2$ are equivalent
  w.r.t. \qcsplogic.
  One can also check that $\mathcal{A}_2$ is minimal and is therefore
  a Q-core of $\mathcal{A}_3$, and \textsl{a posteriori} of
  $\mathcal{A}_4$.
\end{example}
The behaviour of the Q-core differs from its cousins the core and the
$U$-$X$-core.
\begin{proposition}
  The Q-core of a $3$-element structure $\mathcal{A}$ is not always an induced substructure of $\mathcal{A}$.
\end{proposition}
\begin{proof}
  Consider the signature $\sigma:=\langle E,R,G\rangle$ involving a
  binary relation $E$ and two unary relations $R$ and $G$. Let
  $\mathcal{A}$ and $\mathcal{B}$ be structures with domain
  $\{1,2,3\}$ with the following relations. 
  \[ 
  \begin{array}{ccc}
    E^{\mathcal{A}}:=\{(1,1),(2,3),(3,2)\} & R^\mathcal{A}:=\{1,2\} & G^\mathcal{A}:=\{1,3\} \\
    E^{\mathcal{B}}:=\{(1,1),(2,3),(3,2)\} & R^\mathcal{B}:=\{1\} & G^\mathcal{B}:=\{1\}
  \end{array}
  \]
  Since $\mathcal{B}$ is a substructure of $\mathcal{A}$, we have $\mathcal{B} \surhom \mathcal{A}$.
  Conversely, the square of $\mathcal{A}^2$ contains an edge that has
  no vertex in the relation $R$ and $G$, which ensures that $\mathcal{A}^2 \surhom \mathcal{B}$
  (see Figure~\ref{fig:qcores}).
  Observe also that no two-element structure $\mathcal{C}$, and \emph{a
  fortiori} no two-element substructure of $\mathcal{A}$ agrees with
  them on $\qcsplogic$. 
  \begin{figure*}
  \centering
  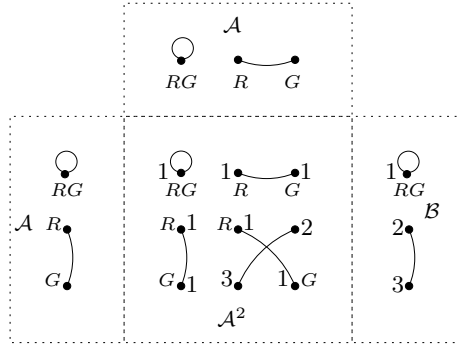  
  \caption{Example of two distinct 3-element structures (signature, $E$ binary and two unary
    predicates $R$ and $G$) that are equivalent
    w.r.t. $\qcsplogic$.}
  \label{fig:qcores}
\end{figure*}
  Indeed, if a structure $\mathcal{C}$ agrees on
  $\qcsplogic$ with $\mathcal{B}$, it agrees also on $\csplogic$. Thus, the core of
  $\mathcal{B}$ is also the core of $\mathcal{C}$ and must appear as an
  induced substructure of $\mathcal{C}$. 
  This core is the one-element substructure of $\mathcal{B}$ induced
  by 1. In order to have a surjective
  homomorphism from a power of $\mathcal{C}$ to $\mathcal{B}$, this
  power must contain a non-loop, and so does $\mathcal{C}$. This
  non-loop must in $\mathcal{C}$ be adjacent to another vertex (this is a
  $\qcsplogic$-expressible property that holds in $\mathcal{B}$
  $\forall x \exists y E(x,y)$). The structure $\mathcal{C}$ would
  therefore be a two element structure satisfying 
  \[ 
  \begin{array}{ccc}
    E^{\mathcal{C}}\subseteq \{(1,1),(1,2),(2,1)\} &
    R^\mathcal{C}\subseteq \{1\}
    & G^\mathcal{C}\subseteq \{1\} \\  
  \end{array}
  \]
  A power of $\mathcal{C}$ would therefore be connected, which is not
  the case of $\mathcal{B}$, preventing the existence of any
  surjective homomorphism.
\end{proof}
We do not know whether the Q-core of a structure is unique. We will explore in the
following section Q-cores over some special classes and show that it behaves well in
these cases. 

\section{Q-cores over classes}

\subsection{The Boolean case}
A Boolean structure $\bB$ has domain $B:=\{0,1\}$. The results of this
section apply to arbitrary (not necessarily finite) signatures. We
give the following lemma ultimately for illustrative purposes (the
gist of its proof will be reused several times). The \emph{pH-type}
$T(b)$ of $b \in \bB$ is the set of all formulae $\phi(x)$ in one free
variable $x$ from $\qcsplogic$ such that $\bB \models \phi(x)$ (pH
stands for positive Horn).
\begin{lemma}
\label{lem:switch-aut}
If $\bB$ is a Boolean structure such that the pH-types $T(0)$ and $T(1)$ in $\bB$ coincide, then $\bB$ has an automorphism swapping $0$ and $1$.
\end{lemma}
\begin{proof}
Suppose there is no such automorphism, then \mbox{w.l.o.g.} we may assume there exists a conjunction of atoms $\theta(x,y)$, involving only variables $x$ and $y$, such that $\bB \models \theta(0,1)$ but $\bB \notmodels \theta(1,0)$. Now, $\bB \models \theta(0,0)$ iff $\bB \models \theta(1,1)$, since $0$ and $1$ are of the same pH-type.  If $\bB \models \theta(0,0)$ then $\bB \models \forall x \theta(x,0)$ but $\bB \notmodels \forall x \theta(x,1)$ (contradiction). Similarly, if $\bB \notmodels \theta(0,0)$ then $\bB \models \exists x \theta(0,x)$ but $\bB \notmodels \exists x \theta(1,x)$ (contradiction).
\end{proof}  
\begin{theorem}
\label{thm:Boolean}
  Let $\bA$ and $\bB$ be Boolean structures that are equivalent
  w.r.t. $\qcsplogic$. Then $\bA$ and $\bB$ are isomorphic.
\end{theorem}
\begin{proof}
We consider the pH-types $T^\bA(0)$ and $T^\bA(1)$ of $0$ and $1$ in $\bA$, respectively (likewise with $\bB$ superscripts for $\bB$).

Case I. $T^\bA(0)=T^\bA(1)$. It follows that $T^\bA(0)=T^\bA(1)=T^\bB(0)=T^\bB(1)$, since $\phi(x) \in T^\bA(0)$ iff $\bA \models \forall x \phi(x)$ iff $\bB \models \forall x \phi(x)$ iff $\theta(x) \in T^\bB(0)$ etc. In this case the function $A \rightarrow B$ given by $0 \mapsto 0$ and $1 \mapsto 1$ is an isomorphism, as in the proof of Lemma~\ref{lem:switch-aut}. (Of course, $0 \mapsto 1$ and $1 \mapsto 0$ is also an isomorphism.)

Case II. $T^\bA(0)$ and $T^\bA(1)$ are incomparable. Let $\theta_0(x)$ be in $T^\bA(0)$ but not in $T^\bA(1)$; and let $\theta_1(x)$ be in $T^\bA(1)$ but not in $T^\bA(0)$. Let $z(0)$ be the witness of $\exists x \theta_0(x)$ in $\bB$; and let $z(1)$ be the witness of $\exists x \theta_1(x)$ in $\bB$. Since $\bA \notmodels \exists x \theta_0(x) \wedge \theta_1(x)$, $\bB \notmodels \exists x \theta_0(x) \wedge \theta_1(x)$ and $z(0)\neq z(1)$. We claim $z$ ($0 \mapsto z(0)$, $1 \mapsto z(1)$) is an isomorphism from $\bA$ to $\bB$. If not, then \mbox{w.l.o.g.} we may assume there exists an atom (maybe with variables identified) $\theta(x,y)$ such that $\bA \models \theta(0,1)$ but $\bB \notmodels \theta(z(0),z(1))$. Now, $\bA \models \theta(0,0)$ iff $\bA \models \theta_0(0) \wedge \theta(0,0)$ iff $\bA \models \exists x \theta_0(x) \wedge \theta(x,x)$ iff $\bB \models \exists x \theta_0(x) \wedge \theta(x,x)$ iff $\bB \models \theta(z(0),z(0))$. The proof is completed as in that of Lemma~\ref{lem:switch-aut}.

Case III. \mbox{W.l.o.g.} $T^\bA(0) \subseteq T^\bA(1)$ but $T^\bA(0) \neq T^\bA(1)$. Let $\theta_1(x)$ be in $T^\bA(1)$ but not in $T^\bA(0)$. Let $z(1)$ be the witness of $\exists x \theta_1(x)$ in $\bB$. $z(1)$ is unique since $\bB \notmodels \forall x \theta_1(x)$. Let $z(0)$ be the other element of $\bB$. We claim $z$ is an isomorphism from $\bA$ to $\bB$. If not, then there exists a conjunction of atoms $\theta(x,y)$ such that (we lose the \mbox{w.l.o.g.}) either $\bA \models \theta(0,1)$ but $\bB \notmodels \theta(z(0),z(1))$, or $\bB \models \theta(0,1)$ but $\bA \notmodels \theta(z(0),z(1))$. In fact, we can still deal with both cases at once, since $\bA \models \theta(1,1)$ iff $\bA \models \theta_1(1) \wedge \theta(1,1)$ iff $\bA \models \exists x \theta_1(x) \wedge \theta(x,x)$ \mbox{etc.} and $\bB \models \theta(1,1)$. The proof concludes as in Lemma~\ref{lem:switch-aut}.
\end{proof}



  In the extended logic $\qcsplogiceq$, it follows that every
  structure of size at most $2$ satisfies the ideal of core. For $\qcsplogic$ we can only say the following. 
\begin{proposition}
Every Boolean structure $\bB$ is either a Q-core, or its Q-core is the substructure induced by either of its elements. In particular, \textbf{the} Q-core of $\bB$ is unique up to isomorphism and is an induced substructure of $\bB$.
\end{proposition}
\begin{proof}
If $\bB$ generates the same QCSP as a one-element structure $\bA$ with domain $\{0\}$, then it is clear that the pH-types $T^\bB(0)$, $T^\bB(1)$ and $T^\bA(0)$ coincide. This gives uniqueness and induced substructure in this case. Otherwise, the only possibility -- to violate the statement -- is for $\bB$ to have as Q-core a non-induced substructure. But this is impossible by Theorem~\ref{thm:Boolean}.
\end{proof}
Lemma~\ref{lem:switch-aut} does not extend to structures of size three. There exists $\bH_1$ of size three such that every element of $\bH_1$ has the same pH-type and yet $\bH_1$ has no non-trivial automorphism ($H_1:=\{0,1,2\}$ and $E^{\bH_1} :=\{(0,0),(1,1),(2,2),(1,2)\}$). However, $\bH_1$ is not a Q-core.


\subsection{Unary structures}

Let $\sigma$ be a fixed relational signature that consists of $n$
unary relation symbols $M_1, M_2,\ldots,M_n$. A structure over such a
signature is deemed \emph{unary}.

Let $w$ be a string of length $n$ over the alphabet $\{0,1\}$. We
write $w(x)$ as an abbreviation for the quantifier-free formula 
$\bigwedge_{1\leq i\leq n, w[i]=1} M_i(x)$.
Each element $a$ of a unary structure $\mathcal{A}$ corresponds to a
word $w$, which is the largest word bitwise such that
$\mathcal{A}\models w(x/a)$.
Let $w_\forall$ be the bitwise $\land$ of the words associated to each element.
The unary structure $\mathcal{A}$ satisfies the \emph{canonical universal sentence}
$\forall y \, w_{\,\forall}(y)$ (note that $w_\forall$ is also the largest word bitwise among such
satisfied universal formulae).  
\begin{proposition}
  The Q-core of a unary structure $\mathcal{A}$ is the unique
  substructure of $\mathcal{A}$ defined as follows.
  The Q-core of $\mathcal{A}$ is the core $\mathcal{A}'$ of
  $\mathcal{A}$ if they share the same canonical universal sentence
  and the disjoint union of $\mathcal{A}'$ with a single element
  corresponding to $w_\forall$ where $\forall y \, w_{\,\forall}(y)$ is the
  canonical universal sentence of $\mathcal{A}$.
\end{proposition}
\begin{proof}
  A positive Horn sentence over a unary signature is logically
  equivalent to a conjunction of formulae that do not share any
  variables. Each conjunct is either a universal formulae of the form
  $\forall y \, w(y)$ or an existential formulae of the form $\exists x
  \, w(x)$. 
  
  The core $\mathcal{A}'$ is a substructure of the Q-core, so we need
  only enforce that the Q-core and $\mathcal{A}'$ satisfy the same
  canonical universal sentence.
  This is achieved with the optional addition of an element corresponding to
  $w_\forall$ where $\forall y w_\forall(y)$ is the canonical
  universal sentence of $\mathcal{A}$. 
\end{proof}

\subsection{Structures with an isolated element}   

We say that a pH-sentence  is a \emph{proper pH-sentence}, if it has at least one universally quantified variable $x_1$ that occurs in some atom in the quantifier-free part.
Generally, we will say that a sentence is \emph{proper} if any variable $x_1$ occurs in some atom and we will be always working with such sentences unless otherwise stated (otherwise, we would simply discard $x_1$ and consider the equivalent sentence with one less variable). We will say that a proper pH-sentence $\psi$ induced from a proper pH-sentence $\phi$ by the removal of some conjuncts is a \emph{proper subsentence} of $\psi$.

Let $\sigma$ be a signature that consists of finitely many relation symbols $R_i$ of respective arity $r_i$. We will consider the set of minimal proper pH-sentences w.r.t. $\sigma$, that is all formulae of the form 
$\forall x_1\exists x_2 \ldots \exists x_{r_i}\, R_i(\bar{x}),$
where the tuple $\bar{x}$ is a permutation  of the variables $x_1,x_2,\ldots x_{r_i}$ where $x_1$ has been transposed with some other variable. There are $r(\sigma)= \Sigma_{R_i \in \sigma} r_i$ such formulae.

\begin{theorem}
  Let $\mathcal{A}$ be a $\sigma$-structure. The following are equivalent.
  \begin{enumerate}%
  \item $\mathcal{A}$ does not satisfy any proper pH-sentence.
  \item $\mathcal{A}$ does not satisfy any of the $r(\sigma)$ minimal proper
    pH-sentences w.r.t. $\sigma$.
  \item $\mathcal{A}^{r(\sigma)}$ contains an isolated element.
  \end{enumerate}
\end{theorem}
\begin{proof}
  The first point implies trivially the second. We show the
  converse. Note that any proper pH-sentence $\phi$ contains as
  a proper subinstance a sentence of the form
  $Q_1 x_1 Q_2 x_2\ldots  \forall x_j \ldots Q_r x_{r}\,
  R(\bar{x})$, where the $Q_i$ represent some arbitrary quantifiers.
  By assumption, the structure $\mathcal{A}$ models $\exists
  x_j \forall x_1 \forall x_2\ldots  \forall x_{r}\, \lnot R(\bar{x})$. 
  Thus, it follows that $\mathcal{A}$ models the weaker sentence $\forall x_1 \forall x_2 \ldots,\exists x_j\ldots  \forall x_{r}$ $\lnot R(\bar{x})$ 
  (the same strategy for selecting a witness for $x_j$ will work)
  and the even weaker sentence where some universal quantifiers are
  turned to existential ones, namely those for which $Q_j$ is
  universal
  (the strategy for these new existential variable can be chosen arbitrarily). 
  So, $\mathcal{A}$ does not model the negation of this last sentence
  which is $Q_1 x_1 Q_2 x_2 \ldots  \forall x_j \ldots Q_r  x_{r}\, R(\bar{x})$, which is a subinstance of $\phi$. By
  monotonicity, $\mathcal{A}$ does not model $\phi$ either.
  
  We now prove that the second point implies the third.
  Let $\phi_i$ be the $i$th minimal proper pH-sentence w.r.t. $\sigma$.
  Let $a_i$ be a witness for the unique existential variable of $\lnot \phi_i$
  that $\mathcal{A}$ does not satisfy $\phi_i$. It is a simple exercise to
  check that $(a_1,a_2,\ldots,a_{r_i})$ is an isolated element of
  $\mathcal{A}^{r(\sigma)}$.
  
  Conversely, if $\bar{a}:=(a_1,a_2,\ldots,a_{r_i})$ is an isolated element of
  $\mathcal{A}^{r(\sigma)}$ then $\bar{a}$ is a witness that $\mathcal{A}^{r(\sigma)}$
  does not satisfy any minimal proper pH-sentence  $\phi_i$. Consequently, there exists some $a_{j_i}$ witnessing that
  $\mathcal{A}$ does not satisfy $\phi_i$ and we are done.
\end{proof}
\begin{example}
  In the case of directed graphs, the minimal proper pH-sentences are $\forall x_1 \exists x_2 E(x_1,x_2)$ and $\forall x_1
  \exists x_2 E(x_2,x_1)$. A directed graph which does not satisfy them
  will satisfy their negation $\exists x_1 \forall x_2 \lnot E(x_1,x_2)$
  and $\exists x_1 \forall x_2 \lnot E(x_2,x_1)$. 
  A witness for the existential $x_1$ in the first
  sentence will be a \emph{source}, and in the second sentence a
  \emph{sink}, respectively

  So a directed graph has a source and a sink if, and only if, it does
  not satisfy any proper pH-sentence, if and only if, its
  square has an isolated element.
\end{example}

\begin{corollary}
  \textbf{The} Q-core of a structure $\mathcal{A}$ that does not satisfy any proper pH-sentence is the unique substructure of $\mathcal{A}$ may be found as
  follows. The Q-core of $\mathcal{A}$ is the core $\mathcal{A}'$ of $\mathcal{A}$, if
  $\mathcal{A}'^{r(\sigma)}$ contains an isolated element, and the
  disjoint union of $\mathcal{A}'$ and an isolated element, otherwise.
\end{corollary}
\begin{proof}
  By assumption $\mathcal{A}$ does not satisfy any proper pH-sentence. Consequently a minimal structure $\mathcal{A}'$ (both w.r.t. domain
  size and number of tuples) which satisfies the same pH-sentences
  as $\mathcal{A}$ will satisfy the same pp-sentences as
  $\mathcal{A}$ and none of the proper pH-sentences
  either. It follows that $\mathcal{A}'$ must contain the core of
  $\mathcal{A}$ (and can be no smaller). If this core satisfies no
  proper pH-sentence then we are done and by the previous
  theorem $\mathcal{A}^{r(\sigma)}$ has an isolated
  element. Otherwise, we must look for a structure that contains the
  core of $\mathcal{A}$ and does not satisfy any proper universal
  sentence. Adding tuples to $\mathcal{A}$ can clearly not force this
  property by monotonicity of QCSP. Thus, the minimal (and unique) such structure
  will be obtained by the addition of an isolated element.
  Note that in this second case we have also a substructure of $\mathcal{A}$.
\end{proof}
\begin{remark}
  It follows that checking whether a structure with an isolated
  element is a Q-core is of the same complexity as checking whether it
  is a core. Recall that the latter is known to be a co-NP-complete
  decision problem (the induced sub-structure that ought to be a core
  is given via an additional monadic predicate $M$), see \cite{HellNesetril}. We show that the
  former is also co-NP-complete (we also assume a monadic predicate as in this
  particular case the Q-core is an induced substructure), by showing
  inter-reducibility of both problem.

  (hardness)
  Let $\langle\mathcal{A},M\rangle$ be the input to the core problem
  (we assume that it is not trivial and that $M$ does not contain any
  isolated element).
  If $\mathcal{A}^{r(\sigma)}$ has an isolated element (it can be
  done in polynomial time as $r(\sigma)$ does not depend on
  $\mathcal{A})$) then we reduce to $\langle\mathcal{A},M
  \rangle$,
  and to $\langle\tilde{\mathcal{A}},\tilde M \rangle$ otherwise, where  
  $\tilde{\mathcal{A}}$ consists of the disjoint union of $\mathcal{A}$ with an isolated
  element and $\tilde M$ is the union of $M$ with this new element.

  (co-NP-complete algorithm)
  Let $\langle\tilde{\mathcal{A}},\tilde M\rangle$ be the input to the
  Q-core problem such that $\tilde{\mathcal{A}}$ has an isolated
  element. 
  We check whether the alleged Q-core (the substructure of
  $\tilde{\mathcal{A}}$ induced by $\tilde{M}$) elevated to the $r$th 
  power has an isolated element. If it does not we answer
  no. If the alleged Q-core has more than one isolated element we
  answer also no.
  Otherwise, we remove at most one isolated element from
  $\tilde M$ to derive $M\subseteq \tilde{M}$ and reduce to the core
  question w.r.t. $\langle\tilde{\mathcal{A}},M\rangle$. 
\end{remark}

\section{The usefulness of Q-cores}

Graphs are relational structures with a single symmetric relation
$E$. We term them \emph{partially reflexive} (p.r.) to emphasise that
any vertex may or may not have a self-loop.
A \emph{p.r. tree} may contain self-loops but no larger cycle $C_n$ for
$n\geq 3$. A \emph{p.r. pseudotree} contains at most one cycle $C_n$ for $n\geq 3$. 
A \emph{p.r. forest} (resp., \emph{pseudoforest}) 
is the disjoint union of p.r. trees (resp., p.r. pseudotrees).

Since \mbox{p.r.} forests  (resp., pseudoforests) 
are closed under substructures, we can be assured that a Q-core of a
\mbox{p.r.} forest (resp., pseudoforest) 
is a \mbox{p.r.} forest (resp., pseudoforest). It is clear from
inspection that the Q-core of \mbox{p.r.} forest (resp., pseudoforest, \mbox{p.r.} cycle) 
is unique up to isomorphism, but we do not prove this as it does not
shed any light on the general situation. The doubting reader may
substitute ``a/ all'' for ``the'' in future references to Q-cores in
this section.

The complexity classifications of~\cite{QCSPforests} were largely
derived using the properties of equivalence w.r.t. $\qcsplogic$.
This will be the central justification for the following propositions.

Let $\bK^*_i$ and $\bK_i$ be the reflexive and irreflexive
$i$-cliques, respectively. Let $[n]:=\{1,\ldots,n\}$. For $i \in [n]$
and $\alpha \in \{0,1\}^n$, let $\alpha[i]$ be the $i$th entry of
$\alpha$. For $\alpha \in \{0,1\}^*$, 
let $\mathcal{P}_{\alpha}$ be
the path with domain $[n]$ and edge set $\{ (i,j) : |j-i|=1 \} \cup \{
(i,i) : \alpha[i]=1 \}$ 
%
For a tree $\mathcal{T}$ and vertex $v \in T$, let $\lambda_T(v)$ be the shortest distance in $\mathcal{T}$ from $v$ to a looped vertex (if $\mathcal{T}$ is irreflexive, then $\lambda_T(v)$ is always infinite). Let $\lambda_T$ be the maximum of $\{\lambda_T(v):v \in T\}$. A tree is \emph{loop-connected} if the self-loops induce a connected subtree. A tree $\mathcal{T}$ is \emph{quasi-loop-connected} if either 1.) it is irreflexive, or 2.) there exists a connected reflexive subtree $\mathcal{T}_0$ (chosen to be \textbf{maximal}) such that there is a walk of length $\lambda_T$ from every vertex of $\mathcal{T}$ to $T_0$. 

\subsection{Partially reflexive forests}

It is not true that, if $\bH$ is a \mbox{p.r.} forest, then either $\bH$ admits a majority polymorphism, and QCSP$(\bH)$ is in NL, or QCSP$(\bH)$ is NP-hard. However, the notion of Q-core restores a clean delineation.
\begin{proposition}
Let $\bH$ be a \mbox{p.r.} forest. Then either the Q-core of $\bH$ admits a majority polymorphism, and QCSP$(\bH)$ is in NL, or QCSP$(\bH)$ is NP-hard.
\end{proposition}
\begin{proof}
  We assume that graphs have at least one edge (otherwise the Q-core is
  $\bK_1$).
  Irreflexive forests are a special case of bipartite graphs, which are
  all equivalent w.r.t. $\qcsplogic$, their Q-core being $\bK_2$ when
  they have no isolated vertex (see
  example~\ref{ex:bip}) and $\bK_2 + \bK_1$ otherwise.

  We assume from now on that graphs have at least one edge and one
  self-loop.
  The one vertex case is $\bK_1^*$. We assume larger graphs from now
  on. If the graph contains an isolated element then its Q-core is $\bK_1 + \bK_1^\star$.
  Assume from now on that the graph does not have an isolated element.

  We deal with the disconnected case first.
  If the graph is reflexive, then its Q-core is $\bK_1^\star +\bK_1^\star$.
  Otherwise, the graph is properly partially reflexive in the sense
  that it embeds both $\bK_1^\star$ and $\bK_1$. If the graph has an
  irreflexive component then its Q-core is $\bK_2 +\bK_1^\star$.
  If the graph has no irreflexive component, then its Q-core is $\bK_1^\star + \bP_{10^\lambda}$ where $\lambda$ is the longest walk from any vertex to a self-loop. To see this last case gives an equivalent QCSP, we may consider power surjective homomorphisms, together with the fact that the Q-core must not satisfy $\forall x \exists y_1,\ldots,y_{\lambda-1} \ E(x,y_1) \wedge E(y_1,y_2) \wedge \ldots \wedge E(y_{\lambda-2},y_{\lambda-1})$.
 
We now follow the classification of \cite{QCSPforests}. If a \mbox{p.r.}
forest contains more than one \mbox{p.r.} tree, then the Q-core is
among those formed from the disjoint union of exactly two (including
the possibility of duplication) of $\bK_1$, $\bK^*_1$, $\bP_{10^\lambda}$,
$\bK_2$. Each of these singularly admits a majority polymorphism, therefore so does any of their disjoint unions. 

We now move on to the connected case, \mbox{i.e.} it remains to consider \mbox{p.r.} trees $\bT$. If $\bT$ is
irreflexive, then its Q-core is $\bK_2$ or $\bK_1$, which admit
majority polymorphisms. If $\bT$ is loop-connected, then it admits a
majority polymorphism \cite{QCSPforests}. If $\bT$ is
quasi-loop-connected, then it is QCSP-equivalent to one of its
subtrees that is loop-connected \cite{QCSPforests} which will be its
Q-core. In all other cases QCSP$(\bT)$ is NP-hard, and $\bT$ does not
admit majority \cite{QCSPforests}. 
\end{proof}

\subsection{Irreflexive Pseudoforests}

A \emph{pseudotree} is a graph that involves at most one cycle. A \emph{pseudoforest} is the disjoint union of a collection of pseudotrees.
\begin{proposition}
Let $\bH$ be an irreflexive pseudoforest. Then either the Q-core of $\bH$ admits a majority polymorphism, and QCSP$(\bH)$ is in NL, or QCSP$(\bH)$ is NP-hard.
\end{proposition}
\begin{proof}
We follow the classification of \cite{CiE2006}. If $\bH$ is bipartite, then its Q-core is either $\bK_2$, $\bK_1$, $\bK_2 + \bK_1$ (see \cite{LICS2008}) and this admits a majority polymorphism. Otherwise its Q-core contains an odd cycle, which does not admit a majority polymorphism, and QCSP$(\bH)$ is NP-hard.
\end{proof}

\section{Computing a Q-core}
We may use Theorem~\ref{theo:containment:mylogic} to provides a first 
algorithm (Algorithm~\ref{algo:naiveQcore}). This does not appear
very promising if we wish to use Q-cores as a preprocessing step.
We will propose and illustrate a general and less naive method to compute Q-cores by
computing $U$-$X$-core and cores first.  

\begin{algorithm}
\SetKwInOut{Input}{input}
\SetKwInOut{Output}{output}
\SetKwInOut{Init}{initialisation}
\SetKw{Guess}{guess}
\SetKw{Set}{set}
\SetKw{KwCheck}{check}
\SetKw{KwInit}{initialisation}

\Input{A structure $\mathcal{A}$}
\Output{The list $L$ of Q-cores of $\mathcal{A}$}
\Init{\Set $L:=\{\mathcal{A}\}$}
\ForAll{substructure $\mathcal{B}$ of $\mathcal{A}$}
{
  \If{there exists a surjective homomorphism from
    $\mathcal{A}^{|A|^{|B|}}$ to $\mathcal{B}$}
  {\If{there exists a surjective homomorphism from $\mathcal{B}^{|B|^{|A|}}$ to
      $\mathcal{A}$} 
    {Remove any structure containing $\mathcal{B}$ in $L$\; 
      Add $\mathcal{B}$ to $L$\;}
  }
}
\Output{List of Q-cores $L$}
\caption{A naive approach to compute the Q-cores.}
\label{algo:naiveQcore}
\end{algorithm}

Another nice feature of cores and $U$-$X$-cores which implies their
uniqueness is the following: any substructure $\mathcal{C}$ of
$\mathcal{A}$ that agrees with it on $\csplogic$ (respectively on
$\mylogic$) will contain the core  (respectively the $U$-$X$-core). 
Consequently, the core and the $U$-$X$-core may be computed in a
\emph{greedy fashion}.
Assuming that the Q-core would not satisfy this nice property,
why should this concern the Q-core?
Well, we know that \textbf{any} Q-core will lie somewhere between the $U$-$X$-core and
the core that are induced substructures: this is a direct consequence of the inclusion of the
corresponding fragments of first-order logic and their uniqueness.
Moreover, according to our current knowledge, checking for equivalence appears, at least on paper, much easier for
$\mylogic$ than $\qcsplogic$: compare the number of functions from $A$
to the power set of $B$ (${2^{|B|}}^{|A|}=2^{|B|\times|A|}$) with the number of functions from
$A^r$ to $B$ ($|B|^{|A|^r}$) where $r$ could be as large as
$|A|^{|B|}$ and can certainly be greater than $r\approx|A|$ (see Proposition~\ref{prop:lowerbound:powerepimorphism}).
So it make sense to bound the naive search for Q-cores.

Furthermore, we know that the $U$-$X$-core can be identified by specific
surjective hypermorphisms that act as the identity on
$X$ and contain the identity on $U$~\cite{LICS2011} which makes the search for the $U$-$X$-core
somewhat easier than its definition suggest (see Algorithm~\ref{algo:UXcore}).

\begin{algorithm}
\SetKwInOut{Input}{input}
\SetKwInOut{Output}{output}
\SetKwInOut{Init}{initialisation}
\SetKwInOut{Variable}{variable}
\SetKwInOut{Variables}{variables}

\SetKw{Guess}{guess}
\SetKw{Set}{set}
\SetKw{KwCheck}{check}
\SetKw{KwLet}{let}

\Input{a structure $\mathcal{A}$.}
\Output{the $U$-$X$-core of $\mathcal{A}$.}
\Variables{$U$ and $X$ two subsets of $A$.}
\Variable{$h$ a surj. hypermorphism from $\mathcal{A}$ to
  $\mathcal{A}$ s.t. $h(U)=A$ and $h^{-1}(X)=A$.}
\Variable{$\mathcal{B}$ an induced substructure of $\mathcal{A}$ such
  that $B=U\cup X$.}

\Init{\Set $U:=A$, $X:=A$, $\mathcal{B}:=\mathcal{A}$, $h$ the identity}

\Repeat{$U$ and $X$ are minimal}{
  \Guess a subset $U'$ of $U$ and a subset $X'$ of $X$\;
  \KwLet $h'$ be a map from $B$ to $B$\;
  \lForAll{$x'$ in $X'$}{\Set $h'(x'):=\{x'\}$}\;
  \lForAll{$u'$ in $U'\setminus X'$}{
    \Guess $x'$ in $X'$
    \Set $h'(u'):=\{u',x'\}$
  }\;
  \ForAll{$z'$ in $B\setminus (U'\cup X')$}
  {
    \Guess $x'$ in $X'$ \Set $h(z'):=\{x'\}$\;
    \Guess $u'$ in $U'$ \Set $h(u'):=h(u')\cup \{z'\}$\;
  }
  \If{$h'$ is a surj. hypermorphism from $B$ to $B$}
  {
    \Set $\mathcal{B}$ to be the substructure of $\mathcal{B}$ induced
    by $U'\cup X'$\; 
    \Set $U:= U'$, $X:= X'$ and $h:= h'\circ h$\;
  }
}
\Output{$\mathcal{B}$.}
\caption{A greedy approach to compute the $U$-$X$-core.}
\label{algo:UXcore}
\end{algorithm}

Observe also that $X$ must contain the core $\mathcal{C}$ of the
$U$-$X$-core $\mathcal{B}$, which is also the core of the original
structure $\mathcal{A}$ (this is because $h$ induces a
so-called retraction of $\mathcal{A}$ to the substructure $\mathcal{A}_{|X}$ induced by
$X$). 
Thus we may compute the core greedily from $X$.
Next, we do a little bit better than using our naive algorithm, by
interleaving steps where we find a substructure that is
$\qcsplogic$-equivalent, with steps where we compute its $U$-$X$-core 
(one can find a sequence of distinct substructures $\mathcal{B}$, $\mathcal{D}$ and
$\mathcal{B}'$ such that $\mathcal{B}$ is a $U$-$X$-core, which is
$\qcsplogic$-equivalent to $\mathcal{D}$, whose $U$-$X$-core
$\mathcal{B}'$ is strictly smaller than $\mathcal{B}$, see Example~\ref{example:interleaving:UXcore:Qretract}).
Algorithm~\ref{algo:BoundedSearchQcore} describes this proposed method
informally when we want to compute \textbf{one} Q-core (of course, we
would have no guarantee that we get the smallest Q-core, unless the
Q-core can be also greedily computed, which holds for all cases we
have studied so far). 

In Algorithm~\ref{algo:BoundedSearchQcore}, we have purposely  not detailed
line~\ref{algo:line:Q-retractionCheck}. We could use the
characterisation of $\qcsplogic$-containment via surjective
homomorphism from a power of Theorem~\ref{theo:containment:qcsplogic} as in Algorithm~\ref{algo:naiveQcore}.
Alternatively, we can use a refined form of (iii) in this Theorem
and use the canonical sentences in $\Pi_2$-form
$\psi_{\mathcal{B},m_1}$ and $\psi_{\mathcal{D},m_2}$, with
$m_1:=\min(|D|,|U|)$ and
$m_2:=|U|$ (see Definition~\ref{def:canonical:pi2:sentence}). The test would consists in checking that $\mathcal{B}$
satisfies $\psi_{\mathcal{D},m_2}$ (where we may relativise to
universal variables to $\mathcal{U}$ and existential variables to $X$) and $\mathcal{D}$ satisfies
$\psi_{\mathcal{B},m_1}$.
This is correct because we know
that we may relativise every universal variable to $U$ within
$\mathcal{B}$. Thus, it suffices to consider $\Pi_2$-sentences with at
most $|U|$ universal variables.

\begin{algorithm}
\SetKwInOut{Input}{input}
\SetKwInOut{Output}{output}
\SetKwInOut{Variable}{variable}
\SetKwInOut{Variables}{variables}
\SetKwInOut{Init}{initialisation}
\SetKw{Guess}{guess}
\SetKw{Pick}{pick}
\SetKw{Set}{set}
\SetKw{KwCheck}{check}

\Input{a structure $\mathcal{A}$}
\Output{a Q-core $\mathcal{B}$ of $\mathcal{A}$}

\Init{
  compute the $U$-$X$-core of $\mathcal{A}$ as in
  Algorithm~\ref{algo:UXcore}\;
  \Set $\mathcal{B}$ to be the $U$-$X$-core\;
  \Set $\mathcal{C}$ to be the core $\mathcal{C}$ of the substructure of
  $\mathcal{B}$ induced by $X$\;} 
\Repeat{$\mathcal{B}$ is minimal}{
  \Guess $\mathcal{D}$ a substructure of $\mathcal{B}$ that contains $\mathcal{C}$\;
  \KwCheck that $\mathcal{B}$ and $\mathcal{C}$ are equivalent
  w.r.t. $\qcsplogic$\;\label{algo:line:Q-retractionCheck}
  \Set $\mathcal{B}$ to be the $U$-$X$-core of $\mathcal{D}$\;
}

\Output{$\mathcal{B}$.}
\caption{Bounded Search for a Q-core.}
\label{algo:BoundedSearchQcore}
\end{algorithm}

\begin{example}\label{example:interleaving:UXcore:Qretract}
  We describe a run of Algorithm~\ref{algo:BoundedSearchQcore} on input
  $\mathcal{A}:=\mathcal{A}_4$.
  During the initialisation, we compute its $U$-$X$-core
  $\mathcal{B}:=\mathcal{A}_3$ and discover that $U=\{2,3\}$ and
  $X=\{1,2\}$. We compute  $\mathcal{C}:=\mathcal{A}_1$, the core of
  the substructure induced by $X$. 

  Note that $\mathcal{B}=\mathcal{A}_3$ is isomorphic to $\mathcal{P}_{110}$.
  Next the algorithm guesses a substructure $\mathcal{D}$ of
  $\mathcal{B}$ that contains $\mathcal{C}$: \textsl{e.g.} it drops the
  self-loop around vertex $3$ to obtain a structure isomorphic to
  $\mathcal{P}_{010}$ and checks successfully equivalence w.r.t. $\qcsplogic$
  (there is a surjective homomorphism from $\mathcal{P}_{010}$ to $\mathcal{P}_{110}$; and,
  conversely we can use the surjective homomorphism from $\mathcal{P}_{110}\times
  \mathcal{P}_{110}$ to $\mathcal{P}_{10}\times
  \mathcal{P}_{11}$ composed with that from the former to $\mathcal{P}_{010}$).

  Next the algorithm computes the $U'$-$X'$-core $\mathcal{B}'$ of $\mathcal{D}$ which is $\mathcal{A}_2$
  (witnessed by $h'(1)=1, h'(2)=h'(3)=\{1,2,3\}$, $U':=\{2\}$,
  $X':=\{1\}$) and sets $\mathcal{B}:=\mathcal{B}'=\mathcal{A}_2$.
  
  The algorithm stops eventually and outputs $\mathcal{A}_2$ as it is minimal.
\end{example}


\section{Conclusion}

We have introduced a notion of Q-core and demonstrated that it does
not enjoy all of the properties of cores and $U$-$X$-core. In
particular, there need not be a unique minimal element \mbox{w.r.t.}
size in the equivalence class of structures agreeing on
pH-sentences. However, we suspect that the notion of Q-core we give
is robust, in that the Q-core of any structure $\bB$ is unique up to
isomorphism; and, that it sits inside any substructure of
$\mathcal{B}$ that satisfies the same sentence of $\qcsplogic$, making
it computable in a greedy fashion. Thus, the nice behaviour of Q-cores is almost restored,
but ``induced substructure'' in the properties of core or $U$-$X$-core
must be replaced by the weaker ``substructure''. 

Generalising the results about Q-cores of structures with an isolated
element to disconnected structures is already difficult. Just as the
pH-theory of structures with an isolated element is essentially
determined by their pp-theory, so the pH-theory of disconnected
structures is essentially determined by its $\forall \exists^*$
fragment (see \cite{CiE2006}). 



\bibliographystyle{acm}

\end{document}

%% file: DifCores4.pdf_t
\begin{picture}(0,0)%
\includegraphics{DifCores4.pdf}%
\end{picture}%
%
%
\setlength{\unitlength}{3108sp}%
\begingroup\makeatletter\ifx\SetFigFont\undefined%
\gdef\SetFigFont#1#2#3#4#5{%
  \reset@font\fontsize{#1}{#2pt}%
  \fontfamily{#3}\fontseries{#4}\fontshape{#5}%
  \selectfont}%
\fi\endgroup%
\begin{picture}(1292,1588)(934,-6554)
\put(1126,-5911){\makebox(0,0)[rb]{\smash{{\SetFigFont{9}{10.8}{\rmdefault}{\mddefault}{\updefault}{\color[rgb]{0,0,0}1}%
}}}}
\put(1575,-6296){\makebox(0,0)[b]{\smash{{\SetFigFont{9}{10.8}{\rmdefault}{\mddefault}{\updefault}{\color[rgb]{0,0,0}3}%
}}}}
\put(2026,-5911){\makebox(0,0)[lb]{\smash{{\SetFigFont{9}{10.8}{\rmdefault}{\mddefault}{\updefault}{\color[rgb]{0,0,0}4}%
}}}}
\put(1576,-5236){\makebox(0,0)[b]{\smash{{\SetFigFont{9}{10.8}{\rmdefault}{\mddefault}{\updefault}{\color[rgb]{0,0,0}2}%
}}}}
\end{picture}%

%% file: DifCores3.pdf_t
\begin{picture}(0,0)%
\includegraphics{DifCores3.pdf}%
\end{picture}%
%
%
\setlength{\unitlength}{3108sp}%
\begingroup\makeatletter\ifx\SetFigFont\undefined%
\gdef\SetFigFont#1#2#3#4#5{%
  \reset@font\fontsize{#1}{#2pt}%
  \fontfamily{#3}\fontseries{#4}\fontshape{#5}%
  \selectfont}%
\fi\endgroup%
\begin{picture}(397,1588)(2542,-6554)
\put(2611,-6296){\makebox(0,0)[b]{\smash{{\SetFigFont{9}{10.8}{\rmdefault}{\mddefault}{\updefault}{\color[rgb]{0,0,0}3}%
}}}}
\put(2611,-5236){\makebox(0,0)[b]{\smash{{\SetFigFont{9}{10.8}{\rmdefault}{\mddefault}{\updefault}{\color[rgb]{0,0,0}2}%
}}}}
\put(2566,-5911){\makebox(0,0)[lb]{\smash{{\SetFigFont{9}{10.8}{\rmdefault}{\mddefault}{\updefault}{\color[rgb]{0,0,0}1}%
}}}}
\end{picture}%

%% file: DifCores2.pdf_t
\begin{picture}(0,0)%
\includegraphics{DifCores2.pdf}%
\end{picture}%
%
%
\setlength{\unitlength}{3108sp}%
\begingroup\makeatletter\ifx\SetFigFont\undefined%
\gdef\SetFigFont#1#2#3#4#5{%
  \reset@font\fontsize{#1}{#2pt}%
  \fontfamily{#3}\fontseries{#4}\fontshape{#5}%
  \selectfont}%
\fi\endgroup%
\begin{picture}(397,960)(3442,-5926)
\put(3466,-5911){\makebox(0,0)[lb]{\smash{{\SetFigFont{9}{10.8}{\rmdefault}{\mddefault}{\updefault}{\color[rgb]{0,0,0}1}%
}}}}
\put(3511,-5236){\makebox(0,0)[b]{\smash{{\SetFigFont{9}{10.8}{\rmdefault}{\mddefault}{\updefault}{\color[rgb]{0,0,0}2}%
}}}}
\end{picture}%

%% file: DifCores1.pdf_t
\begin{picture}(0,0)%
\includegraphics{DifCores1.pdf}%
\end{picture}%
%
%
\setlength{\unitlength}{3108sp}%
\begingroup\makeatletter\ifx\SetFigFont\undefined%
\gdef\SetFigFont#1#2#3#4#5{%
  \reset@font\fontsize{#1}{#2pt}%
  \fontfamily{#3}\fontseries{#4}\fontshape{#5}%
  \selectfont}%
\fi\endgroup%
\begin{picture}(388,373)(4351,-5926)
\put(4366,-5911){\makebox(0,0)[lb]{\smash{{\SetFigFont{9}{10.8}{\rmdefault}{\mddefault}{\updefault}{\color[rgb]{0,0,0}1}%
}}}}
\end{picture}%

%% file: DifCores2Square.pdf_t
\begin{picture}(0,0)%
\includegraphics{DifCores2Square.pdf}%
\end{picture}%
%
%
\setlength{\unitlength}{3108sp}%
\begingroup\makeatletter\ifx\SetFigFont\undefined%
\gdef\SetFigFont#1#2#3#4#5{%
  \reset@font\fontsize{#1}{#2pt}%
  \fontfamily{#3}\fontseries{#4}\fontshape{#5}%
  \selectfont}%
\fi\endgroup%
\begin{picture}(1599,1599)(2914,-6373)
\put(3691,-6271){\makebox(0,0)[lb]{\smash{{\SetFigFont{9}{10.8}{\rmdefault}{\mddefault}{\updefault}{\color[rgb]{0,0,0}$\mathcal{A}_2$}%
}}}}
\put(3151,-5416){\makebox(0,0)[rb]{\smash{{\SetFigFont{9}{10.8}{\rmdefault}{\mddefault}{\updefault}{\color[rgb]{0,0,0}$\mathcal{A}_2$}%
}}}}
\end{picture}%

%% file: DifCores2squareTo3.pdf_t
\begin{picture}(0,0)%
\includegraphics{DifCores2squareTo3.pdf}%
\end{picture}%
%
%
\setlength{\unitlength}{3108sp}%
\begingroup\makeatletter\ifx\SetFigFont\undefined%
\gdef\SetFigFont#1#2#3#4#5{%
  \reset@font\fontsize{#1}{#2pt}%
  \fontfamily{#3}\fontseries{#4}\fontshape{#5}%
  \selectfont}%
\fi\endgroup%
\begin{picture}(998,915)(3445,-6241)
\put(3466,-6001){\makebox(0,0)[lb]{\smash{{\SetFigFont{9}{10.8}{\rmdefault}{\mddefault}{\updefault}{\color[rgb]{0,0,0}1}%
}}}}
\put(3466,-5551){\makebox(0,0)[lb]{\smash{{\SetFigFont{9}{10.8}{\rmdefault}{\mddefault}{\updefault}{\color[rgb]{0,0,0}3}%
}}}}
\put(4321,-6001){\makebox(0,0)[lb]{\smash{{\SetFigFont{9}{10.8}{\rmdefault}{\mddefault}{\updefault}{\color[rgb]{0,0,0}1}%
}}}}
\put(4276,-5596){\makebox(0,0)[lb]{\smash{{\SetFigFont{9}{10.8}{\rmdefault}{\mddefault}{\updefault}{\color[rgb]{0,0,0}2}%
}}}}
\end{picture}%

%% file: QcoresNotInduced.pdf_t
\begin{picture}(0,0)%
\includegraphics{QcoresNotInduced.pdf}%
\end{picture}%
%
%
\setlength{\unitlength}{3108sp}%
\begingroup\makeatletter\ifx\SetFigFont\undefined%
\gdef\SetFigFont#1#2#3#4#5{%
  \reset@font\fontsize{#1}{#2pt}%
  \fontfamily{#3}\fontseries{#4}\fontshape{#5}%
  \selectfont}%
\fi\endgroup%
\begin{picture}(3624,2724)(-2261,-1423)
\put(901,-241){\makebox(0,0)[b]{\smash{{\SetFigFont{7}{8.4}{\rmdefault}{\mddefault}{\updefault}{\color[rgb]{0,0,0}$RG$}%
}}}}
\put(-1799,-241){\makebox(0,0)[b]{\smash{{\SetFigFont{7}{8.4}{\rmdefault}{\mddefault}{\updefault}{\color[rgb]{0,0,0}$RG$}%
}}}}
\put(-494,614){\makebox(0,0)[lb]{\smash{{\SetFigFont{7}{8.4}{\rmdefault}{\mddefault}{\updefault}{\color[rgb]{0,0,0}$R$}%
}}}}
\put( 46,614){\makebox(0,0)[rb]{\smash{{\SetFigFont{7}{8.4}{\rmdefault}{\mddefault}{\updefault}{\color[rgb]{0,0,0}$G$}%
}}}}
\put(-899,614){\makebox(0,0)[b]{\smash{{\SetFigFont{7}{8.4}{\rmdefault}{\mddefault}{\updefault}{\color[rgb]{0,0,0}$RG$}%
}}}}
\put(-2069,-511){\makebox(0,0)[rb]{\smash{{\SetFigFont{8}{9.6}{\rmdefault}{\mddefault}{\updefault}{\color[rgb]{0,0,0}$\mathcal{A}$}%
}}}}
\put(-494,1064){\makebox(0,0)[b]{\smash{{\SetFigFont{8}{9.6}{\rmdefault}{\mddefault}{\updefault}{\color[rgb]{0,0,0}$\mathcal{A}$}%
}}}}
\put(-899,-241){\makebox(0,0)[b]{\smash{{\SetFigFont{7}{8.4}{\rmdefault}{\mddefault}{\updefault}{\color[rgb]{0,0,0}$RG$}%
}}}}
\put(-494,-241){\makebox(0,0)[lb]{\smash{{\SetFigFont{7}{8.4}{\rmdefault}{\mddefault}{\updefault}{\color[rgb]{0,0,0}$R$}%
}}}}
\put( 46,-241){\makebox(0,0)[rb]{\smash{{\SetFigFont{7}{8.4}{\rmdefault}{\mddefault}{\updefault}{\color[rgb]{0,0,0}$G$}%
}}}}
\put(-1844,-961){\makebox(0,0)[rb]{\smash{{\SetFigFont{7}{8.4}{\rmdefault}{\mddefault}{\updefault}{\color[rgb]{0,0,0}$G$}%
}}}}
\put(-1844,-511){\makebox(0,0)[rb]{\smash{{\SetFigFont{7}{8.4}{\rmdefault}{\mddefault}{\updefault}{\color[rgb]{0,0,0}$R$}%
}}}}
\put(-944,-511){\makebox(0,0)[rb]{\smash{{\SetFigFont{7}{8.4}{\rmdefault}{\mddefault}{\updefault}{\color[rgb]{0,0,0}$R$}%
}}}}
\put(-944,-961){\makebox(0,0)[rb]{\smash{{\SetFigFont{7}{8.4}{\rmdefault}{\mddefault}{\updefault}{\color[rgb]{0,0,0}$G$}%
}}}}
\put(1081,-421){\makebox(0,0)[b]{\smash{{\SetFigFont{8}{9.6}{\rmdefault}{\mddefault}{\updefault}{\color[rgb]{0,0,0}$\mathcal{B}$}%
}}}}
\put(-494,-1276){\makebox(0,0)[b]{\smash{{\SetFigFont{8}{9.6}{\rmdefault}{\mddefault}{\updefault}{\color[rgb]{0,0,0}$\mathcal{A}^2$}%
}}}}
\put( 46,-556){\makebox(0,0)[lb]{\smash{{\SetFigFont{9}{10.8}{\rmdefault}{\mddefault}{\updefault}{\color[rgb]{0,0,0}2}%
}}}}
\put(-584,-961){\makebox(0,0)[lb]{\smash{{\SetFigFont{9}{10.8}{\rmdefault}{\mddefault}{\updefault}{\color[rgb]{0,0,0}3}%
}}}}
\put(766,-556){\makebox(0,0)[lb]{\smash{{\SetFigFont{9}{10.8}{\rmdefault}{\mddefault}{\updefault}{\color[rgb]{0,0,0}2}%
}}}}
\put(766,-1006){\makebox(0,0)[lb]{\smash{{\SetFigFont{9}{10.8}{\rmdefault}{\mddefault}{\updefault}{\color[rgb]{0,0,0}3}%
}}}}
\put(811,-106){\makebox(0,0)[rb]{\smash{{\SetFigFont{9}{10.8}{\rmdefault}{\mddefault}{\updefault}{\color[rgb]{0,0,0}1}%
}}}}
\put(-989,-106){\makebox(0,0)[rb]{\smash{{\SetFigFont{9}{10.8}{\rmdefault}{\mddefault}{\updefault}{\color[rgb]{0,0,0}1}%
}}}}
\put(-494,-106){\makebox(0,0)[rb]{\smash{{\SetFigFont{9}{10.8}{\rmdefault}{\mddefault}{\updefault}{\color[rgb]{0,0,0}1}%
}}}}
\put(136,-106){\makebox(0,0)[rb]{\smash{{\SetFigFont{9}{10.8}{\rmdefault}{\mddefault}{\updefault}{\color[rgb]{0,0,0}1}%
}}}}
\put(-764,-511){\makebox(0,0)[rb]{\smash{{\SetFigFont{9}{10.8}{\rmdefault}{\mddefault}{\updefault}{\color[rgb]{0,0,0}1}%
}}}}
\put(-764,-1006){\makebox(0,0)[rb]{\smash{{\SetFigFont{9}{10.8}{\rmdefault}{\mddefault}{\updefault}{\color[rgb]{0,0,0}1}%
}}}}
\put(-314,-511){\makebox(0,0)[rb]{\smash{{\SetFigFont{9}{10.8}{\rmdefault}{\mddefault}{\updefault}{\color[rgb]{0,0,0}1}%
}}}}
\put(-44,-961){\makebox(0,0)[rb]{\smash{{\SetFigFont{9}{10.8}{\rmdefault}{\mddefault}{\updefault}{\color[rgb]{0,0,0}1}%
}}}}
\put( 46,-961){\makebox(0,0)[lb]{\smash{{\SetFigFont{7}{8.4}{\rmdefault}{\mddefault}{\updefault}{\color[rgb]{0,0,0}$G$}%
}}}}
\put(-494,-511){\makebox(0,0)[rb]{\smash{{\SetFigFont{7}{8.4}{\rmdefault}{\mddefault}{\updefault}{\color[rgb]{0,0,0}$R$}%
}}}}
\end{picture}%